\documentclass[conference,status=final]{IEEEtran}

\usepackage{graphicx}
\usepackage{url}
\usepackage{hyperref}
\usepackage{acro}

\usepackage{amsmath}
\usepackage{amssymb}
\usepackage{amsthm}
\usepackage[greek,english]{babel}
\usepackage{stmaryrd}
\usepackage{eufrak}
\usepackage{pifont}
\usepackage{array}
\usepackage{alltt}
\usepackage{bold-extra}
\usepackage{graphicx}
\usepackage{bussproofs}
\usepackage{xspace}
\usepackage{bm}
\usepackage{relsize}
\usepackage{url}
\usepackage{xcolor}
\usepackage{ifthen}

\newtheorem{theorem}{Theorem}[section]%

\theoremstyle{definition}
\newtheorem{definition}[theorem]{Definition}%

\newcommand{\conditional}[3]{#1 \mathop{\hbox{\raisebox{0ex}{\larger$\leftslice$}}} #2 \mathop{\hbox{\raisebox{0ex}{\larger$\rightslice$}}} #3}%
\newcommand{\hoaretriple}[3]{\left\{ #1 \middle\}\, #2 \,\middle\{ #3 \right\}} %

\renewcommand{\vec}[1]{\bm{#1}}
\newcommand{\ssp}[2][]{\ifthenelse{\equal{#1}{}}{}{#1 +} \left[#2\right]}

\hypersetup{
    colorlinks,
    linkcolor={red!50!black},
    citecolor={blue!50!black},
    urlcolor={blue!80!black}
}

\newcommand{\isaref}[1]{\href{#1}{\includegraphics[width=9pt]{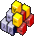}}}
\newcommand{\isalink}[1]{\hfill \isaref{#1}}
\newcommand{\norm}[1]{\left\lVert#1\right\rVert}

\def\land{\mathrel{\wedge}}

\usepackage{subfig}
\usepackage{tikz}
\usetikzlibrary{calc,automata}
\usepackage{cleveref}
\DeclareAcronym{AAIP}{short = AAIP, long = Assuring Autonomy International Programme}
\DeclareAcronym{AMV}{short = AMV, long = Autonomous Marine Vehicle}
\DeclareAcronym{AP}{short = AP, long = autopilot}
\DeclareAcronym{CAM}{short = CAM, long = Collision Avoidance Mode}
\DeclareAcronym{HCM}{short = HCM, long = High Caution Mode}
\DeclareAcronym{LRE}{short = LRE, long = Last Response Engine}
\DeclareAcronym{MOM}{short = MOM, long = Main Operating Mode}
\DeclareAcronym{OCM}{short = MOM, long = Operator Control Mode}
\DeclareAcronym{oCC}{short = oCC, long = {on-Collision-Course}}
\DeclareAcronym{nO}{short = nO, long = {near-Obstacle}}
\DeclareAcronym{CAS}{short = CAS, long = computer algebra systems}
\DeclareAcronym{NC}{short = NC, long = numerical computation}

\newcommand{\dL}{$\textsf{d}\mathcal{L}$\xspace}
\newcommand{\dH}{$\textsf{d}\mathcal{H}$\xspace}

\newcommand{\isabelledH}{Isabelle/\dH}

\begin{document}
\title{Towards Deductive Verification of Control \\ Algorithms for Autonomous Marine Vehicles}
\author{
  \author{\IEEEauthorblockN{
      Simon Foster\IEEEauthorrefmark{2},
      Mario Gleirscher\IEEEauthorrefmark{1}\IEEEauthorrefmark{2},
      Radu Calinescu\IEEEauthorrefmark{1}\IEEEauthorrefmark{2}}
    \IEEEauthorblockA{\IEEEauthorrefmark{2}\textit{Department of
        Computer Science, University of York}, York, UK}
    \IEEEauthorblockA{\IEEEauthorrefmark{1}\textit{Assuring Autonomy
        International Programme, University of York}, York, UK\\ 
      simon.foster,mario.gleirscher,radu.calinescu@york.ac.uk}
  }}

\maketitle  

\begin{abstract}
The use of autonomous vehicles in real-world applications is often precluded by the difficulty of providing safety guarantees for their complex controllers. The simulation-based testing of these controllers cannot deliver sufficient safety guarantees, and the use of formal verification is very challenging due to the hybrid nature of the autonomous vehicles. 
Our work-in-progress paper introduces a formal verification approach that addresses this challenge by integrating the \acl{NC} of such a system (in GNU/Octave) with its hybrid system verification by means of a proof assistant (Isabelle). To show the effectiveness of our approach, we use it to verify differential invariants of an \acl{AMV} with a controller switching between multiple modes.
\end{abstract}


\section{Introduction}
\label{sec:introduction}

Engineering controllers for autonomous 
vehicles requires a range of models, e.g.~of the dynamics and
of the control algorithms, for validating and verifying their key properties~\cite{Luckcuck2018,Gleirscher2018-NewOpportunitiesIntegrated}.  \Acl{NC}~(\acs{NC}, e.g. with MATLAB) 
is a widely used simulation technique for model validation~(i.e.~closing the reality gap) and controller testing.  However, simulation is, like testing, mostly limited to the
demonstration of defects, since it can only consider a small fraction of the input space. 
For correctness, particularly to assess safety, full coverage of this space is desirable or mandatory.  %
For hybrid systems, full coverage can be achieved only using symbolic reasoning techniques, such as deductive verification~\cite{Platzer2008}, due to the uncountable state space. We therefore need the translation of a validated model into a form amenable to verification in a proof environment such as Isabelle/HOL~\cite{Isabelle}.

In this work, we investigate this translation for the case of a hybrid
model of an \acf{AMV} and the formal verification of its safety properties.
We describe the dynamics of the vehicle's motion, and controllers for
waypoint approach and obstacle avoidance.  We model the controller
using hybrid state charts,
including the mode switching for
mitigating accidents between the operator and the safety
controller.  We simulate our model in the \ac{NC} tool GNU/Octave,\footnote{GNU/Octave. \url{http://octave.sourceforge.io/}} for the purpose of
validation against real-world trials, and translate this into an
implementation of differential Dynamic Logic~\cite{Platzer2008,Foster2020-dL}~(\dL) in Isabelle/HOL for
deductive verification. To support this, we extend it to support matrices, discrete state, and a form of modular verification.

Our preliminary work serves as a template for how a translation from an \ac{NC} tool to Isabelle can be achieved, and provides additional evidence that Isabelle provides a credible and flexible solution for hybrid systems verification.
Our work is inspired by Mitsch et al.~\cite{Mitsch2017Obstacle} who provide a generic verified model for collision avoidance in KeYmaera~X~\cite{KeYmaeraX}. 
We advance their work through provision of explicit support for transcendental functions in the system dynamics, a higher-level notation in our tool that bridges the semantic gap with control engineers, and access to Isabelle's automated proof facilities.

After an overview of the technologies we use in
\Cref{sec:preliminaries}, we present our approach to validation-based
formal verification in \Cref{sec:approach} and close with a discussion
in \Cref{sec:conclusion}.

\section{Background}
\label{sec:preliminaries}

Isabelle/HOL~\cite{Isabelle} is a proof assistant for Higher Order Logic (HOL). It includes a functional specification language and an array of proof facilities, including \textit{sledgehammer} ~\cite{Blanchette2011}, which integrates automated provers, such as Z3. Isabelle is highly extensible, and has a variety of mathematical libraries, notably for Multivariate Analysis~\cite{Harrison2005-Euclidean} and Ordinary Differential Equations~\cite{Immler2012,Immler2014} (ODEs), which provide the foundations for verification of hybrid systems. 

Isabelle/UTP~\cite{Foster2020-IsabelleUTP} is a semantic framework based on Hoare and He's Unifying Theories of Programming (UTP)~\cite{Hoare&98}, built on Isabelle/HOL. It supports diverse semantic models in a variety of paradigms, such as reactive, concurrent, and hybrid systems, and their application to verification. For example, it contains a tactic, \textit{hoare-auto}, that automates verification of sequential programs using Hoare logic that by utilising \textit{sledgehammer} to discharge verification conditions.

\dL is a logic for deductive verification of hybrid systems, which is supported by the KeYmaera X tool~\cite{KeYmaeraX}. \dL includes a hybrid program modelling language, and a verification calculus based on dynamic logic. It can be used to prove invariants both of control algorithms and continuous dynamics, which makes it ideal for verifying hybrid systems. It avoids the need for explicit solutions to differential equations, by using a technique called \emph{differential induction}. Recently, differential induction has been embedded into Isabelle~\cite{Munive2018-DDL} and Isabelle/UTP~\cite{Foster2020-dL} to create differential Hoare logic (\dH), which also supports verification of hybrid programs, but in a more general setting. In this paper, we integrate \dH into Isabelle/UTP, and extend it.

\section{Approach}
\label{sec:approach}

Our case study, the C-Worker~5\footnote{C-Worker
  5. \url{https://www.asvglobal.com/product/c-worker-5/}}~(\Cref{fig:cworker5})
is an \ac{AMV} designed to support hydrographic survey work. It
operates in the open sea and so must avoid collisions with both static
and dynamic obstacles, such as rocky outcrops and other vessels.
We consider a safety controller that (1)
avoids collisions with obstacles where possible by taking evasive
maneuvers; and (2) mitigates the effects where
avoidance is impossible. Our industrial partner, D-RisQ\footnote{D-RisQ Software Systems. \url{http://www.drisq.com/}}, is developing a safety controller called the \ac{LRE}~\cite{Foster2020AUV} implementing the above
functionality when the boat is operating autonomously.
For verification, we focus on avoidance of static obstacles. 

\subsection{Modelling the Dynamics and the Controller}
\label{sec:casmodel}

\subsubsection*{Modelling the \ac{AMV} Dynamics}

The dynamical model should be close enough to reality to do \ac{NC} and abstract enough to reduce the complexity of formal
verification to a level appropriate for a credible assurance case~\cite{Gleirscher2019-SEFM}.

\begin{figure}
  \subfloat[\ac{AMV} C-Worker 5]{
    \includegraphics[width=.4\linewidth]{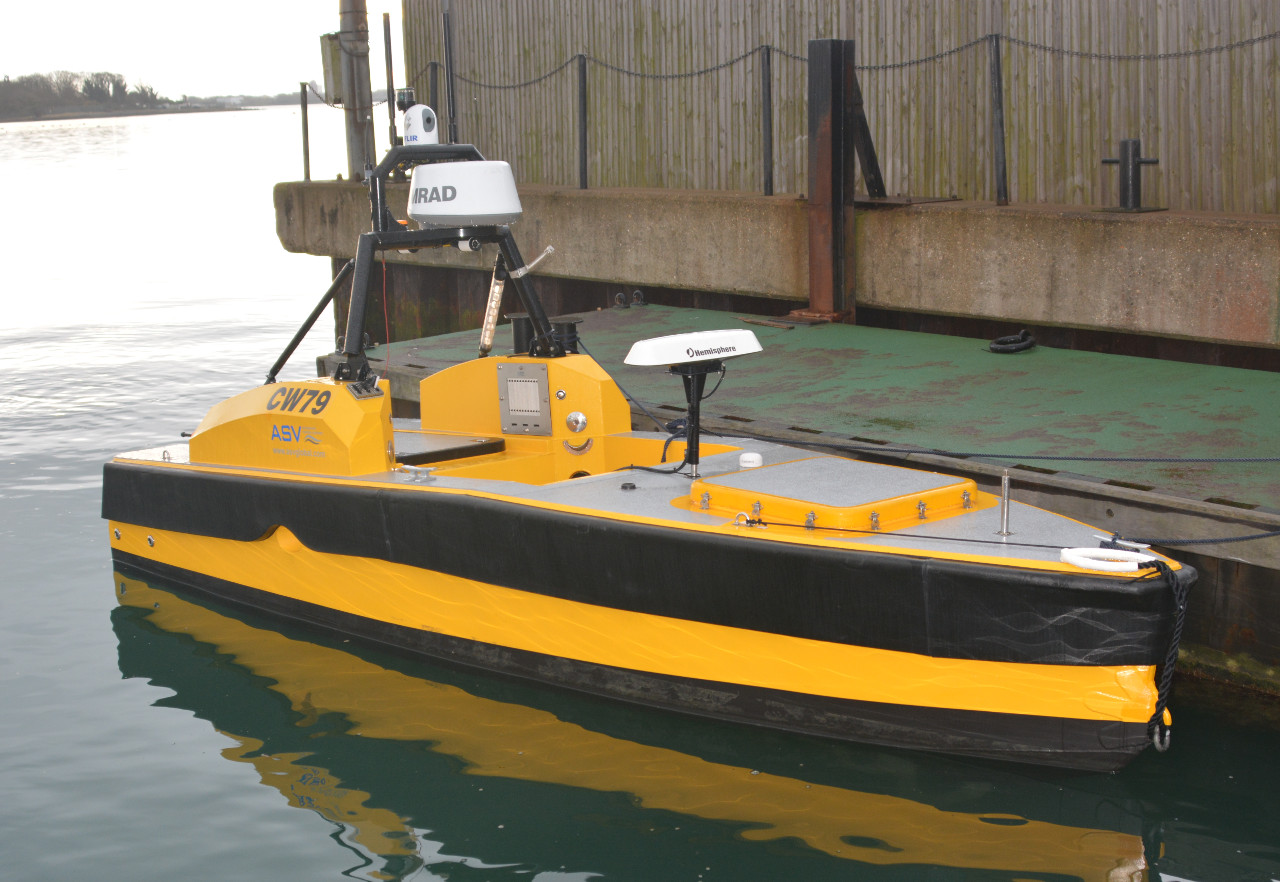}
    \label{fig:cworker5}}
  \hfill
  \subfloat[\ac{AMV} physics]{
\tikzset{
  boat/.pic={
    \draw[thick,blue,fill=blue!10]
    (0,0) -- ++(-1,0) -- ++(0,2) arc (180:90:1cm and 2cm)
    arc (90:0:1cm and 2cm) -- ++(0,-2) -- (0,0);
  },
  vfor/.pic={
    \draw[vec] (0,0) -- (0,2.5)
    node[left,rotate=30] {$\vec{v}$}; 
    \path (0,0) rectangle (1.5,1.5) node (vforcorn) {};
    \draw[vec,color=green!30!black!50] (0,0) -- (0,1.5)
    node[pos=.8,left,rotate=30] {$\vec{f}_l$}; 
    \draw[vec,color=green!30!black!50] (0,0) -- (1.5,0)
    node[pos=.8,below,rotate=30] {$\vec{f}_t$}; 
    \draw[vec,color=green!30!black!50] (0,0) -- (vforcorn)
    node[pos=.8,below,rotate=30] {$\vec{f}$}; 
    \draw (.25,0) arc (0:90:.25);
  }
}
\begin{tikzpicture}
  [every node/.style={transform shape},
  vec/.style={thick,black,-stealth},
  scale=.7,>=stealth]
  \useasboundingbox (-.5,-1.6) rectangle (4.9,2);

  \draw[->,thick] (-1.5,1.2) edge node[left] {y} ++(0,1)
  (-1.5,1.2) -> node[below] {x} ++(1,0);
  \draw[very thick,blue,->] (-.4,-1.5) .. controls (0,1.6) and (2.3,3)
  .. (4.7,-.7) node[pos=.95,sloped,above] {trajectory};
  \draw[dashed,gray] (-.5,0) -- (5,0);
  
  \draw[fill=black] (0,0) circle (2pt) node (A) {};
  \node[anchor=north east] at (A.south west) {A};
  \draw[fill=black] (5,-1) circle (2pt) node (W) {};
  \node[anchor=north east] at (W.south east) {W};
  \draw[fill=black] (3,2) circle (2pt) node (O) {};
  \node[anchor=south west] at (O.south east) {O};
  \draw (0,0) pic[rotate=-30,scale=.4] {boat};
  \draw (0,0) pic[rotate=-30,scale=1] {vfor};
  \draw[right] (1,0) arc (0:60:1) node[pos=.5] {$\phi_A$};

  \draw[vec,gray] (A) edge node[pos=.5] (phiao) {}
  node[pos=.8,above] {$\overline{AO}$} (O);
  \draw[right] (phiao) arc (30:60:1.6) node[pos=.5] {$\phi_{AO}$};
  \draw[vec,gray] (A) edge node[pos=.1] (phiaw) {}
  node[pos=.6,below] {$\overline{AW}$} (W);
  \draw[right] (phiaw) arc (-12:60:.6) node[pos=0,right] {$\phi_{AW}$};
\end{tikzpicture}

     \label{fig:model:physics}}
  \caption{An \ac{AMV} in real and a model of its physics
    \label{fig:amv}}
  \vspace{-2ex}
\end{figure}
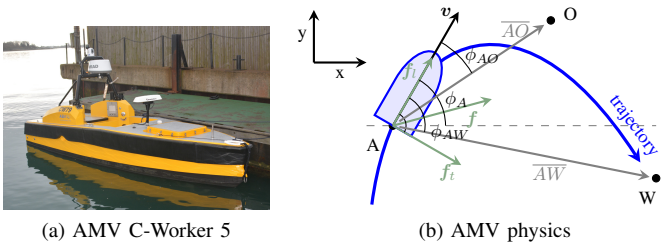

Indicated in \Cref{fig:model:physics}, at time $t \in T$, we consider
the velocity $\vec{v}_A = [v_A^x,v_A^y]^T$ and position $\vec{p}_A$ of
the \ac{AMV}, the position $\vec{p}_W$ of a next waypoint to be
approached, and a set $O$ of obstacles, each described by its velocity
$\vec{v}_{O_i}$ and position $\vec{p}_{O_i}$.  $\vec{p}$ and $\vec{v}$
are vectors in planar coordinates $(x,y)$ over
$\mathbb{R}^2$.
These parameters form a state space $\mathcal{X}$ with tuples
\[
  \vec{x} = [\vec{p}_A,\vec{v}_A,\vec{p}_O,\vec{v}_O]^T.
\]
Below, we abbreviate $\vec{p}_E$ by $E$ where $E\in\{A,O,W\}$.  We
also consider parameters calculated from $\vec{x}$, such as the
distance to the next waypoint $\norm{\overline{AW}}$ or the angle
$\phi_{AO}$ between the \ac{AMV} velocity vector and the distance
vector $\overline{AO}$.

For sake of simplicity, we consider the \ac{AMV} as a particle with
mass $m$ and formulate its dynamics as the following system of
ordinary differential equations
\begin{align}
  \label{eq:model:odes}
  \dot{\vec{p}}_A = \vec{v},\quad
    \dot{\vec{v}}_A = \vec{f} / m,\quad
  \dot{\vec{v}}_O = \vec{0},
  \quad\text{and}\quad
    \dot{\vec{p}}_O = \vec{0}
\end{align}
where $\vec{f} / m$ implements Newton's second law of the kinetics of
particle masses relating a force applied to the vehicle and this
vehicle's acceleration at time $t$.  To remain in scope of our
investigation, we further simplify the \ac{AMV} dynamics, 
omitting disturbances~(e.g.~crosswind) and perturbations~(e.g.~flow
resistance), and restricting our analysis to static obstacles.

\subsubsection*{Modelling the \ac{AMV} Controller}

\Cref{fig:model:overall} shows the structure of the plant consisting
of the dynamical model of the \ac{AMV} and its environment~(as
explained before) and a two-layered controller comprising the \ac{AP}
and the \ac{LRE}.

\begin{figure}[t]
  \centering
\begin{tikzpicture}
  [tra/.style={draw,minimum height=2em,align=center},
  sum/.style={circle,draw},
  every node/.style={transform shape},
  spl/.style={circle,black,fill=black},
  lab/.style={align=center},
  scale=.7,>=stealth]
  \node[sum] (lremon) at (0,0) {};
  \node[lab,anchor=south] (lremonl) at (lremon.north) {LRE\\Mon.};
  \node[lab,anchor=north west] at (lremon.south east) {$\mathbb{B}$};
  \node[tra] (lre) at ($(lremon)+(1.5,0)$) {LRE};
  \node[sum] (apmon) at ($(lre)+(1.5,0)$) {};
  \node[tra] (ap) at ($(apmon)+(1.5,0)$) {AP};
  \node[lab,anchor=south] (apmonl) at (apmon.north) {AP\\Mon.};
  \node[lab,anchor=north west] at (apmon.south east) {$-$};
  \node[tra] (dc) at ($(ap)+(1.5,0)$) {D/C};
  \node[tra] (dyn) at ($(dc)+(2,0)$) {AMV Dyn.};
  \draw[spl] ($(dyn)+(1.5,0)$) circle (.2em);
  \node (dynout) at ($(dyn)+(1.5,0)$) {};
  \node[tra] (envsim) at ($(dc)+(-1,1)$) {Env. Sim.};
  \node[tra] (env) at ($(envsim)+(3,0)$) {Env. Dyn.};
  \draw[spl] ($(env)+(2,0)$) circle (.2em);
  \node (envout) at ($(env)+(2,0)$) {};
  \node[tra] (cd) at ($(dyn)+(0,-1)$) {C/D};
  \node[tra] (out) at ($(cd)+(-2,0)$) {Aggregator};
  \draw[thick,->] ($(lremon)+(-1,0)$) -- (lremon);
  \draw[thick,->] (lremon) edge (lre)
  (lre) edge[above] node {$\vec{w}$} (apmon)
  (apmon) edge (ap)
  (ap) edge[above] node {$\vec{f}$} (dc)
  (dc) edge (dyn)
  (cd) edge (out)
  (dynout) edge ($(dynout)+(1,0)$);
  \draw (dyn) -- (dynout);
  \draw[thick,->,above left] (dynout) |- node {$\vec{x}$} (cd);
  \draw[thick,->,below,near start] (out) -| node {$\vec{x,y}$} (apmon);
  \draw[thick,->] (out) -| (lremon);
  \draw[thick,->] ($(envsim)+(-1.5,0)$) -- (envsim);
  \draw[thick,->] (envsim) edge (env);
  \draw[thick] (env) -- (envout);
  \draw[thick,->] (envout) edge ($(envout)+(0.5,0)$);
  \draw[thick,->] (envout) |- (cd);
\end{tikzpicture}

   \caption{Block structure diagram of the dynamical model and the
    two-layered controller with the corresponding monitors
  \label{fig:model:overall}}
  \vspace{-2ex}
\end{figure}
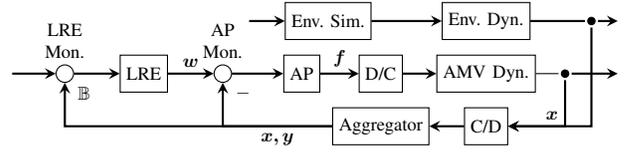

The \emph{discrete low-level control} of the vehicle is facilitated by
the \ac{AP} through generating the propulsive force $\vec{f}$ of the
\ac{AMV} as an input to the \ac{AMV} dynamics.  Within the frame of
reference of the trajectory of the \ac{AMV}, we model the \ac{AMV}'s
single thruster by calculating two components of $\vec{f}$, the
longitudinal (or tangential) acceleration force $\vec{f}_l$ collinear
with the \ac{AMV}'s velocity $\vec{v}$ and the radial acceleration
force $\vec{f}_t$ perpendicular to $\vec{f}_l$, such that
\begin{align*}
  \vec{f} = \vec{f}_l + \vec{f}_t =
  f_l 
  \left[
  \begin{array}{c}
    \cos(\phi_{A})\\
    \sin(\phi_{A})
  \end{array}
  \right]
  + f_t 
  \left[
  \begin{array}{c}
    -\mathrm{sgn}(\phi_{A})\sin(\phi_{A})\\
    \mathrm{sgn}(\phi_{A})\cos(\phi_{A})
  \end{array}
  \right].
\end{align*}
The \emph{discrete high-level control} of the \ac{AMV} is partially
facilitated by the \ac{LRE} through switching between several
operating modes: an \ac{OCM}, a \ac{MOM}, a \ac{HCM}, and a \ac{CAM}. When in OCM, the operator has responsibility for the AMV. When in MOM, the AMV navigates towards the next waypoint at maximum speed. If it gets close to, but not on collision course with, an obstacle then it switches to HCM. If a potential future collision is detected, it transitions to CAM to make evasive maneuvers.
Each of these modes provides the \ac{AP} with a particular
setpoint $\vec{w} = [\mathit{rs}, \vec{p}_W]^T$~(i.e.~target
speed and location of next waypoint) for the calculation of $\vec{f}$
by the \ac{AP} as described by the hybrid automaton in
\Cref{fig:model:ha}
\begin{align}
  \label{eq:model:dyn:flong}
  f_l &= k_p^l \cdot |\mathit{rs} - \norm{\vec{v}_A}|
  \\
  \label{eq:model:dyn:frad}
  f_t &=
  \left\{
  \begin{array}{ll}
    k_p^t \cdot \phi_{AW},
    & \text{in MOM} \\
    k_p^t \cdot \mathrm{sgn}(\phi_{AW}) \cdot f_{max},
    & \text{in CAM/HCM}
  \end{array}
  \right.
\end{align}
In this example, we use a simple proportional controller for $\vec{f}$
with directional proportionality factors $k_p$ as shown in
\Cref{eq:model:dyn:flong,eq:model:dyn:flong}.
For obstacle avoidance manoeuvres, we calculate the \emph{safe braking
  distance} by
\begin{align}
  \label{eq:model:sbd}
  d_{sb} = \frac{sb \cdot \norm{\vec{v}_A}^2 \cdot m}{-2 \cdot k_p^b
  \cdot f_{max}}
\end{align}
with a safety margin $sb$ to capture modelling uncertainty and define
the \emph{\ac{nO}} and \emph{\ac{oCC}} hazards as the predicates
\begin{align}
  \label{eq:model:occ}
  \mathit{nO} \equiv \norm{\overline{AO}} > d_{sb}
  \quad\text{and}\quad
  \mathit{oCC} \equiv \mathit{nO}
  \land |\phi_{AO}| < \epsilon_{\phi}.
\end{align}
In \texttt{MOM}, we add a \emph{hysteresis} $\epsilon_{h}$ to
$\epsilon_{\phi}$ in order to delay manoeuvre cancellation.
\Cref{fig:model:ha} describes the overall behaviour of how the
\ac{LRE} switches between the four modes to provide $\vec{w}$ to the
\ac{AP}.  From the components \ac{LRE} and \ac{AP} shown in
\Cref{fig:model:overall} and from the modes shown in
\Cref{fig:model:ha}, one can then derive the interfaces for the
detailed software design of the \ac{AMV} control system.

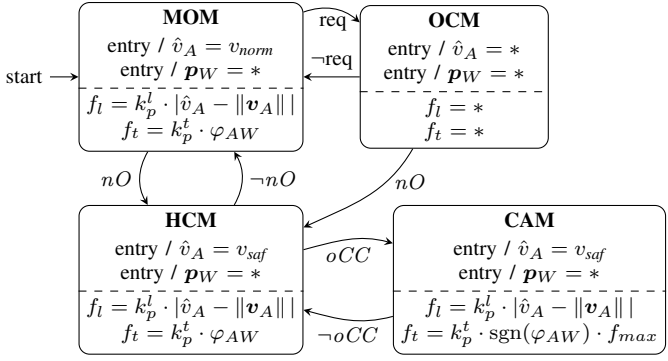
\begin{figure}
  \centering
\footnotesize
\begin{tikzpicture}
  [node distance=3.5cm,
  state/.style={rounded corners,minimum width=2.5cm,minimum height=1.5cm,draw,align=center},
  every node/.style={transform shape},
  scale=1,>=stealth]

  \node[state,initial] (mom) {\textbf{MOM}
    \\[.2em] entry / $\hat{v}_A = v_{\textit{norm}}$
    \\entry / $\vec{p}_W = *$
    \\[.5em] $f_l = k_p^l \cdot |\hat{v}_A - \norm{\vec{v}_A}|$
    \\$f_t = k_p^t \cdot \varphi_{AW}$};
  \node[state] (ocm) [right of=mom] {\textbf{OCM}
    \\[.2em] entry / $\hat{v}_A = *$
    \\entry / $\vec{p}_W = *$
    \\[.5em] $f_l=*$
    \\$f_t=*$
  };
  \node[state] (hcm) at ($(mom)+(0,-2.7cm)$) {\textbf{HCM}
    \\[.2em] entry / $\hat{v}_A = v_{\textit{saf}}$
    \\entry / $\vec{p}_W = *$
    \\[.5em] $f_l = k_p^l \cdot |\hat{v}_A - \norm{\vec{v}_A}|$
    \\$f_t = k_p^t \cdot \varphi_{AW}$
  };
  \node[state] (cam) at ($(hcm)+(4.5,0)$) {\textbf{CAM}
    \\[.2em] entry / $\hat{v}_A = v_{\textit{saf}}$
    \\entry / $\vec{p}_W = *$
    \\[.5em] $f_l = k_p^l \cdot |\hat{v}_A - \norm{\vec{v}_A}|$
    \\$f_t = k_p^t \cdot \mathrm{sgn}(\varphi_{AW}) \cdot f_{max}$
  };
  \draw[dashed] ($(mom.west)+(0,-.15)$) -- ($(mom.east)+(0,-.15)$);
  \draw[dashed] ($(ocm.west)+(0,-.7em)$) -- ($(ocm.east)+(0,-.7em)$);
  \draw[dashed] ($(cam.west)+(0,-.15)$) -- ($(cam.east)+(0,-.15)$);
  \draw[dashed] ($(hcm.west)+(0,-.15)$) -- ($(hcm.east)+(0,-.15)$);
  \draw[->] (mom) edge[bend left,below] node {req} (ocm)
  (ocm) edge[above] node {$\neg$req} (mom)
  (mom) edge[bend right,left] node {$\mathit{nO}$} (hcm)
  (hcm) edge[bend right,right] node {$\neg\mathit{nO}$} (mom)
  (ocm) edge[out=240,in=25,right] node[pos=.3] {$\mathit{nO}$} (hcm)
  (hcm) edge[out=15,in=165,below] node {$\mathit{oCC}$} (cam)
  (cam) edge[out=195,in=-15,below] node {$\neg \mathit{oCC}$} (hcm);
\end{tikzpicture}

  \caption{Behaviour of the \ac{LRE} and \ac{AP} as a Moore machine
    $*$\dots non-deterministic assignment by the operator
    \label{fig:model:ha}}
  \vspace{-2ex}
\end{figure}

\subsubsection*{Note on Abstraction}

The transition from the discrete \ac{LRE} and \ac{AP} to the
continuous \ac{AMV} physics is accomplished by a conversion of
(D)iscretely timed $\vec{f}$ inputs in form of (C)ontinuous,
piece-wise constant signals, processed by actuators.  Vice versa, the
digital controller~(particularly, the Aggregator in
\Cref{fig:model:overall}) samples the environment through sensors at a
certain rate.  \Cref{fig:model:overall} indicates this abstraction by
D/C and C/D converters.  Although we chose to apply this abstraction
to the generation of $\vec{f}$, in practice, this will happen inside
the thrusters where, e.g.~digital signals control a servo motor of a
combustion engine and a rudder to generate $\vec{f}$.

\subsubsection*{Simulating the Model}

We implemented the \ac{AMV} model in a simple integrator-based
simulator in plain GNU/Octave.  For that, we derived parameters, such as
weight, maximum speed and propulsive force, from the C-Worker~5
specification.  Additionally, we identified controller constants, such
as $k_p^l$, during simulation.
\Cref{fig:sim} shows a trajectory of the \ac{AMV}~(green dot) turning around to
make its way to the next waypoint $W$~(blue dot) while circumventing
a floating obstacle~(pink dot).  The initial state
$\vec{x}_0\in\mathcal{X}$ is set to
\[
  \vec{x}_0 = [
  \underbrace{-.5,-3.8}_{\vec{v}_A [m/s]},
  \underbrace{-10,-10}_{\vec{p}_A [m]},
  \underbrace{0,0}_{\vec{v}_{O_1} [m/s]},
  \underbrace{-12,-18}_{\vec{p}_{O_1} [m]},\dots
  ]^T
\]
and the simulation run for the constants $\mathit{rs}=4 m/s$,
$\vec{p}_W = [0,0]^T$, and the time interval $T=[0,35]$ sec.  Note, the 2D
trajectory from the \ac{AMV} exhibits a deviation from its
course where $\mathit{oCC}$ turned true.  The lower middle graph shows
this as the event of $d_{sb}\approx 13 > |\overline{AO}|$ where the
magenta curve touches the red curve at $t\approx 5s$, simultaneous to
the reduction of $\norm{\vec{v}_A}$ after the switch to
$\mathit{CAM}$~(cf.~top right graph).

\begin{figure}
  \includegraphics[width=\columnwidth]{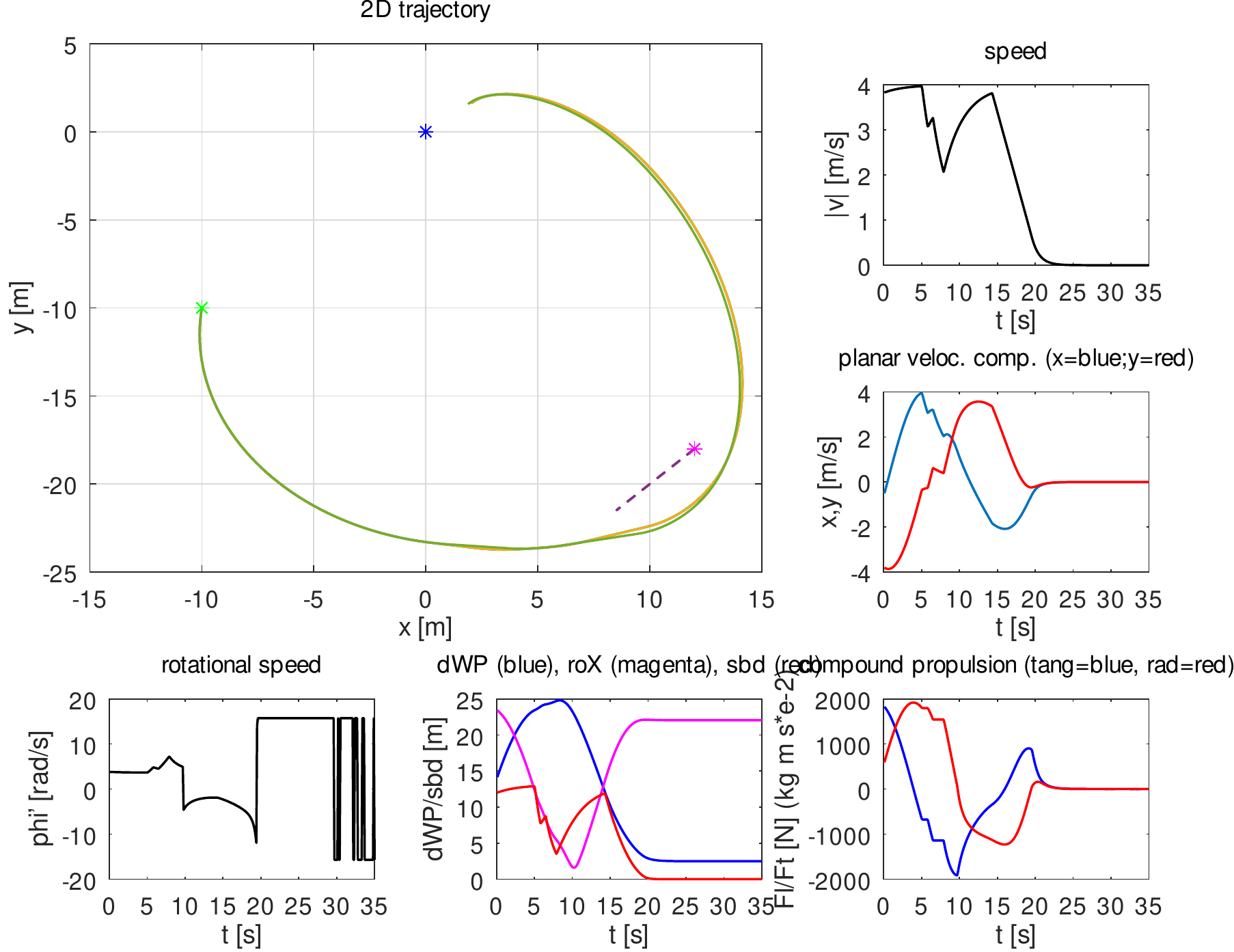}
  \caption{Simulation: the \ac{AMV}~(green dot) approaching next
    waypoint~(blue dot) while crossing an obstacle~(pink dot)
    \label{fig:sim}}
  \vspace{-3ex}
\end{figure}

\subsubsection*{Beyond Simulation}

Our quest for covering the input space~(\Cref{sec:introduction}) requires us to ask how we can know that from wherever in $\mathcal{X}$ we start,
wherever an obstacle is, in whatever interval $T$ we evaluate a
trajectory, will the \ac{AMV} always steer away from an obstacle in
\ac{CAM}, will it always reduce speed in \ac{HCM}, will it reach the
next waypoint within a given time in \ac{MOM}?  Such questions require
a more fundamental investigation of the model discussed in the next
section.

\subsection{Verification}

In order to support deductive verification of \acp{AMV}, we apply Isabelle/UTP to prove properties of the controller and system dynamics. For this, we utilise \isabelledH which integrates \dH~\cite{Munive2018-DDL,Foster2020-dL} into Isabelle/UTP. We extend it with matrices, discrete variables, and modular reasoning. Along with standard Hoare logic laws, \isabelledH includes the key rules from \dL, including differential induction and cut, which are at the core of our verification approach. We prove these laws as theorems of Hoare logic in Isabelle/UTP\footnote{We omit the proofs for reasons of space. They can be found in our repository (\url{https://github.com/isabelle-utp/utp-main}) and accompanying ~\includegraphics[width=9pt]{isabelle_transparent}~ links.}.

\begin{theorem}{Differential Induction and Cut} \label{thm:dI} \isalink{https://github.com/isabelle-utp/utp-main/blob/cfcac3847198564f09a99931a74543d8ca64ce11/theories/hyprog/utp_hyprog_dinv.thy\#L303}
\begin{align}
& ~~~ \dfrac{\textit{differentiable}(P) \land (B \Rightarrow
\mathcal{L}_{F}(P))}{\hoaretriple{P}{\langle \dot{\vec{x}} = F(\vec{x}) ~|~ B(\vec{x})\rangle}{P}}
\label{eq:dI}\\[1ex]
& \dfrac{\hoaretriple{P}{\langle F ~|~ B \rangle}{P} ~ \hoaretriple{Q}{\langle F ~|~ B \land P \rangle}{Q}}{\hoaretriple{P \land Q}{\langle F ~|~ B \rangle}{Q}} \label{eq:dC}
\end{align}

\end{theorem}

\noindent Here, $\langle \dot{\vec{x}} = F(\vec{x}) ~|~ B(\vec{x})\rangle$ is a system of ODEs with an evolution domain $B$. The dynamical system is permitted to evolve provided that the ODEs in $F$, and predicate $B$, are satisfied for all points on the solution trajectory. \eqref{eq:dI} states that if $P$ is everywhere differentiable, and its differentiated form follows from $B$, then $P$ is an invariant. \eqref{eq:dC} shows that if we can prove that $P$ is invariant, then we can use it as an axiom of the dynamics to prove that $Q$ is also an invariant~\cite{Platzer2008}. 

We extend \cite{Foster2020-dL} with support for automation of Lie derivatives~\cite{Ghorbal2014-AlgInv} evaluation.
The ODEs are encoded as a vector field $F : \mathbb{R}^n \to \mathbb{R}^n$. $\mathcal{L}_F(P)$ denotes the Lie derivative of the predicate $P$ along $F$. $P$ is restricted to the form $e~R~f$, for $R \in \{=, \le, <\}$, over differentiable expressions $e, f : \mathbb{R}^n \to \mathbb{R}$, and their conjunctions and disjunctions, such as $2x \le 5$. We exemplify $\mathcal{L}$ below. \isalink{https://github.com/isabelle-utp/utp-main/blob/2152a303f6a517f125a44c7709010b00e6e2b0a8/theories/hyprog/utp_hyprog_deriv.thy\#L113}
\begin{align}
\mathcal{L}_F(e \le f) &= (\mathcal{L}_F(e) \le \mathcal{L}_F(f)) \\
\mathcal{L}_F(e + f)   &= \mathcal{L}_F(e) + \mathcal{L}_F(f) \\
\mathcal{L}_F(e \cdot f)   &= \mathcal{L}_F(e) \cdot f + e \cdot \mathcal{L}_F(f) \label{eq:prodrule} \\
\mathcal{L}_F(\sin(e))  &= \mathcal{L}_F(e) \cdot \cos(e) \\
\mathcal{L}_F(x)       &= (\lambda s : \mathbb{R}^n.\, \textit{\textsf{get}}_x~F(s)) \label{eq:lenscvar} %
\end{align}
\noindent These are largely standard, such as the product rule~\eqref{eq:prodrule}. Of note, \eqref{eq:lenscvar} shows the treatment of a continuous variable $x : \mathbb{R}$. We encode mutable variables using lenses~\cite{Foster07}, which are pairs $$(x : V \Longrightarrow S) \triangleq (\textit{\textsf{get}}_x : S \to V, \textit{\textsf{put}}_x : S \to V \to S)$$ for some suitable state space $S$ and variable type $V$, that obey intuitive algebraic laws~\cite{Foster2020-IsabelleUTP}. Here, we require that every continuous variable has a bounded linear $\textit{\textsf{get}}$ function, which is satisfied, for example, when $x$ is a projection of a Euclidean space. The derivative of $x$ is an expression that applies the $\textit{\textsf{get}}$ function to the derivative of the state ($F(s)$). This can be seen as a semantic substitution of $x$ by its derivative~\cite{Foster2020-IsabelleUTP}. %

Invariants can contain transcendental functions such as $\sin$ and $log$. We also support equalities between arbitrary Euclidean spaces, 
such as matrices. To close the gap between Octave and Isabelle, we have implemented a smart matrix parser. A matrix in Isabelle is represented by a function: $A~\textit{mat}[M, N] \triangleq N \to M \to A$, where $N$ and $M$ are finite types denoting the dimensions, and $A$ is the element type, usually $\mathbb{R}$. Thus, the Isabelle type system can be used to ensure that matrix expressions are well-formed. We use the syntax $$\left[[x^1_1, x^1_2, \cdots x^1_n], \cdots, [x^m_1, x^m_2, \cdots, x^m_n]\right]$$ to represent a $m$ by $n$ matrix in Isabelle, which is a list of lists. Our parser can infer the dimensions of a well-formed matrix, and produce suitable dimension types, which aids proof. Moreover, we have proved theorems that allow symbolic evaluation of certain vector operations, for example:
\begin{align*}
    [x_1, y_1] + [x_2, y_2] &= [x_1 + x_2, y_1 + y_2] \\
    n \cdot [x, y] &= [n \cdot x, n \cdot y]
\end{align*}
\noindent We also define the matrix lens \isalink{https://github.com/isabelle-utp/utp-main/blob/cfcac3847198564f09a99931a74543d8ca64ce11/theories/hyprog/utp_hyprog_prelim.thy\#L51} $$\textit{mat-lens}(m : M, n : N) : A \Longrightarrow A~\textit{mat}[M, N]$$ which accesses an element. With it, we can model both variables that refer to an entire matrix and also its elements.

We have developed a tactic in \isabelledH called $\textit{dInduct}$, which automates the application of Theorem~\ref{thm:dI} by determining whether $P$ is indeed differentiable everywhere, and if so applying differentiation and substitution. The resulting predicate can be discharged, or refuted, using Isabelle's tactics.

Hybrid systems in \isabelledH follow the pattern of $Sys \triangleq (Ctrl \mathrel{\hbox{\rm;}} Dyn)^\star$, where the controller and dynamics iteratively take turns in updating the variables~\cite{Mitsch2017Obstacle}. Proving a safety property $P$ of $Sys$ entails finding an invariant $I$ both of $Ctrl$ and $Dyn$, such that $I \Rightarrow P$. \isabelledH splits the state space of a hybrid system into its continuous and discrete variables. %
Continuous variables change during evolution,
but discrete variables are constant and updated only by assignments.

The continuous state space ($\Sigma_C$) must form a Euclidean space, and so is typically composed of reals, vectors, and matrices. There are no restrictions on the discrete state space, and it may use any Isabelle data type. 
\subsubsection*{Case Study}
We describe each of the continuous variables using lenses, e.g. $\vec{p}, \vec{v} : \mathbb{R}~\textit{mat}[1,2] \Longrightarrow \Sigma_C$\footnote{We omit the $A$ subscripts for brevity.}. In addition to those mentioned in \S\ref{sec:casmodel}, we also include $\vec{a}: \mathbb{R}~\textit{mat}[1,2]$ for the acceleration, and $s : \mathbb{R}$, for the linear speed. Technically, $s$ can be derived as $\norm{v}$, but its inclusion makes proving invariants easier. Most are monitored variables of the environment, except $\vec{a}$, which is updated by the \ac{AP}. The discrete variables include
waypoint location ($\vec{wp} : \mathbb{R}^2$);
obstacle set ($ob : \mathbb{P}(\mathbb{R}^2)$);
linear speed and heading set points ($rs, rh : \mathbb{R}$);
force vector ($\vec{f} : \mathbb{R}^2$); and mode ($m: \{\textit{OCM}, \textit{MOM}, \textit{HCM}, \textit{CAM}\}$).
Next, we describe the dynamics.

\begin{definition}[AMV Dynamics] \label{def:amv-dyn} \isalink{https://github.com/isabelle-utp/utp-main/blob/2152a303f6a517f125a44c7709010b00e6e2b0a8/theories/hyprog/examples/AMV.thy\#L257}
\begin{align*}
    dyn_{AV} &\triangleq
    \left(\begin{array}{l}
        \dot{t} = 1;\ \dot{\vec{p}} = \vec{v};\ \dot{\vec{v}} = \vec{a};\ \dot{\vec{a}} = 0;\\
        \dot{s} = \conditional{\dfrac{\vec{v} \cdot \vec{a}}{s}}{s \neq 0}{\norm{\vec{a}}}; \\
        \dot{\phi} = \conditional{acos\left(\dfrac{(\vec{v} + \vec{a}) \cdot \vec{v}}{\norm{\vec{v} + \vec{a}} \cdot \norm{\vec{v}}}\right)}{s \neq 0}{0} 
    \end{array}\right) \\[1ex]
    ax_{AV} &\triangleq
        \left(\begin{array}{l}
            0 \le s \land s \le S \land s \cdot 
            \begin{bmatrix} \sin(\phi) \\ \cos(\phi) \end{bmatrix} = \vec{v} \land t < \epsilon
        \end{array}\right) \\[1ex]
    Dyn &\triangleq t := 0 \mathrel{\hbox{\rm;}} \left\langle dyn_{AV} ~|~ ax_{AV} \right\rangle
\end{align*}
\end{definition}

\noindent Dynamics $dyn_{AV}$ is a system of six ODEs. For the linear speed derivative, we consider the special case when $s = 0$, where the speed derivative is derived from $\vec{a}$, and the rotational speed is 0. We also axiomatise some properties of the dynamics in $ax_{AV}$: (1) the linear speed must be in $[0, S]$; (2) $\vec{v}$ must be the same as $s$ multiplied by the orientation unit vector; (3) time must not advance beyond $\epsilon$, which puts an upper bound on the time between control decisions. 

Next, we model the LRE, which is encoded as a set of guarded commands derived from Figure~\ref{fig:model:ha}. In each iteration, the LRE updates state variables and can transition to a different state. Whilst in MOM, the speed set point is the maximum speed (S), and the LRE invokes the command \textit{steerToWP} that updates the heading towards the current way point. 

\begin{definition}[Simplified LRE] $LRE \triangleq$ \isalink{https://github.com/isabelle-utp/utp-main/blob/2152a303f6a517f125a44c7709010b00e6e2b0a8/theories/hyprog/examples/AMV.thy\#L278}
$$
\left(\begin{array}{l}
m = \textit{MOM} \rightarrow 
    \left(\!\begin{array}{l}
        rs := S \mathrel{\hbox{\rm;}} \textit{steerToWP} \mathrel{\hbox{\rm;}} \\
        \textbf{if}~\textit{oCC} ~ \textbf{then}~ m := \textit{CAM}~\textbf{fi} \mathrel{\hbox{\rm;}} \\
        \textbf{if}~\exists \vec{o} \in ob. \norm{\vec{o} - \vec{p}} \le D \\
        \textbf{then}~ m := \textit{HCM} \mathrel{\hbox{\rm;}} rs := H ~\textbf{fi}
    \end{array}\right) \\[3ex]
m = \textit{HCM} \rightarrow 
    \left(\begin{array}{l}
        rs := H \mathrel{\hbox{\rm;}} \textit{steerToWP} \mathrel{\hbox{\rm;}} \\
        \textbf{if}~\forall \vec{o} \in ob. \norm{\vec{o} - \vec{p}} > D \\
        \textbf{then}~ m := \textit{MOM} ~\textbf{fi}
    \end{array}\right) \\[3ex]
m = \textit{OCM} \rightarrow \textbf{skip} \\
m = \textit{CAM} \rightarrow \cdots
\end{array}\!\right)
$$
\end{definition}

\noindent If \ac{oCC} is detected the LRE transitions to CAM. If 
\ac{nO} holds but not \ac{oCC},
then the LRE switches to HCM. From HCM, the speed set point is decreased to
$H$. Once the AMV is no longer close to an obstacle, the LRE may return to MOM. OCM exhibits no behaviour since the operator provides the control inputs. Finally, CAM is where collision avoidance procedures are executed. Its behaviour is left unspecified for now. 
The final component we model is the \acl{AP}.

\begin{definition}[Autopilot Controller] $AP \triangleq$ \isalink{https://github.com/isabelle-utp/utp-main/blob/2152a303f6a517f125a44c7709010b00e6e2b0a8/theories/hyprog/examples/AMV.thy\#L291}
$$
\begin{array}{l}
\textbf{if} \norm{rs - s} > s_\epsilon \\
\textbf{then}~ ft := \mathrm{sgn}(rs - s) \cdot \min\left(\!\begin{array}{c}kp_{gv} \cdot \norm{rs - s}, \\ f_{max}\end{array}\!\right) \\
\textbf{else}~ ft := 0 ~\textbf{fi} \mathrel{\hbox{\rm;}} \\
\textbf{if} \norm{rh - \phi} > \phi_\epsilon \\
\textbf{then}~ fl := \mathrm{sgn}(rh - \phi) \cdot \min\left(\!\begin{array}{c}kp_{gr} \cdot \norm{rh - \phi}, \\ f_{max}\end{array}\!\right) \\
\textbf{else}~ fl := 0 ~\textbf{fi} \mathrel{\hbox{\rm;}} \\
\vec{f} := 
    fl \cdot 
    \begin{bmatrix} 
        \cos(\phi) \\ 
        \sin(\phi)
    \end{bmatrix} +
    ft \cdot 
    \begin{bmatrix} 
        \sin(\phi) \\ 
        \cos(\phi)
    \end{bmatrix} \mathrel{\hbox{\rm;}} \vec{a} := \vec{f} / m 
\end{array}
$$

\end{definition}

\noindent The \ac{AP} takes $rs$ and $rh$ as inputs, computes $\vec{f}$, and calculates $\vec{a}$. The constants $s_{\epsilon}$ and $\phi_{\epsilon}$ limit the controller activity when the speed is close to the set point. Its representation in Isabelle/UTP is shown in Figure~\ref{fig:AP-Isabelle}.
Continuous variables are distinguished using the namespace $\mathbf{c}$, e.g.~$\textbf{c}{:}x$.
Scalar multiplication and division are distinguished operators, $n *_R \vec{x}$ and $n \mathop{/_R} \vec{x}$. Finally, we describe the overall AMV behaviour.

\begin{figure}
    \centering
    \includegraphics[width=\linewidth]{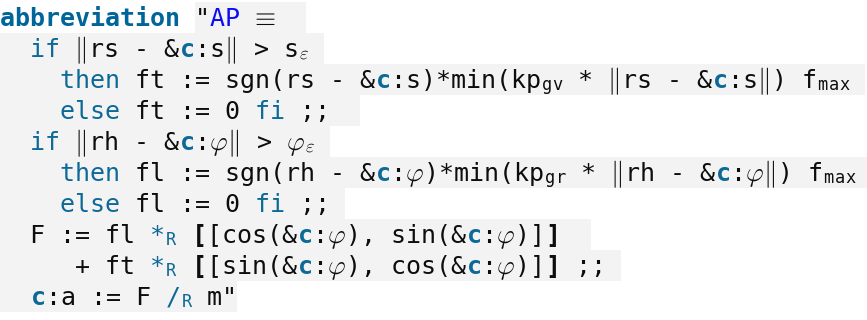}

    \vspace{-1ex}
    \caption{Autopilot in Isabelle/UTP}
    \label{fig:AP-Isabelle}
    
    \vspace{-3ex}
\end{figure}

\begin{definition} $AMV \triangleq (LRE \mathrel{\hbox{\rm;}} AP \mathrel{\hbox{\rm;}} Dyn)^\star$
\end{definition}

\noindent The LRE executes first to determine the new speed set points. Following this, the autopilot calculates the new acceleration vector. Finally, the dynamical system evolves the continuous variables for up to $\epsilon$ seconds, and then the cycle begins again. We will now proceed to verify some properties of the system using \isabelledH. We begin with some structural properties.

\begin{theorem}[Structural Properties] \label{thm:structp} $ $ \isalink{https://github.com/isabelle-utp/utp-main/blob/2152a303f6a517f125a44c7709010b00e6e2b0a8/theories/hyprog/examples/AMV.thy\#L308}
\begin{itemize}
    \item $LRE \mathop{\,\textit{\textbf{nmods}}\,} \{t, \vec{p}, \vec{v}, \vec{a}, s, \phi\}$
    \item $AP \mathop{\,\textit{\textbf{nmods}}\,} \{t, \vec{p}, \vec{v}, s, \phi\}$
    \item $Dyn \mathop{\,\textit{\textbf{nmods}}\,} \{\vec{wp}, ob, rs, rh, ft, fl, \vec{f}, m\}$
\end{itemize}
\end{theorem}

\noindent $P \mathop{\,\textit{\textbf{nmods}}\,} A$ means that $P$ does not modify the variables in $A$. It enables modular verification using the following theorem:
$$S \mathop{\,\textit{\textbf{nmods}}\,} \vec{x} \implies \hoaretriple{p(\vec{x})}{S}{p(\vec{x})}$$ If $\vec{x}$ is not modified by $S$, then any predicate in $\vec{x}$ is invariant. We can verify that the LRE does not modify any of the continuous variables, as it only updates the (discrete) set points. The \ac{AP} modifies only the continuous variable $\vec{a}$. We can prove that $ax_{AV}$ is an invariant of both $LRE$ and $AP$, and therefore of the entire system. The dynamics can potentially change any of the continuous variables, but does not change any of the discrete variables. These structural properties are automatically proved, and are useful to ensure structural well-formedness of a controller under development.

We next prove some invariants of the system using \dH.

\begin{theorem}[Collinearity of $\vec{v}$ and $\vec{a}$] \isalink{https://github.com/isabelle-utp/utp-main/blob/2152a303f6a517f125a44c7709010b00e6e2b0a8/theories/hyprog/examples/AMV.thy\#L452} \label{thm:collinear}
$$\hoaretriple{\vec{a} \cdot \vec{v} = \norm{\vec{a}} \cdot \norm{\vec{v}}}{Dyn}{\vec{a} \cdot \vec{v} = \norm{\vec{a}} \cdot \norm{\vec{v}}}$$
\end{theorem}
\begin{proof}
We first prove that $\vec{a} \cdot \vec{v} \ge 0$ and $(\vec{a} \cdot \vec{v})^2 = (\vec{a} \cdot \vec{a}) \cdot (\vec{v} \cdot \vec{v})$ are both invariants by \cref{thm:dI}. We can then show these are equivalent with $a \cdot v = \norm{\vec{a}} \cdot \norm{\vec{v}}$.
\end{proof}
\noindent Collinearity means that $\vec{v}$ and $\vec{a}$ have the same direction and the AMV is travelling straight. A corollary is below. \isalink{https://github.com/isabelle-utp/utp-main/blob/2152a303f6a517f125a44c7709010b00e6e2b0a8/theories/hyprog/examples/AMV.thy\#L487}
$$\hoaretriple{\!
    \vec{a} \cdot \vec{v} = \norm{\vec{a}}\!\cdot\!\norm{\vec{v}}\!
    }{Dyn}{\vec{p} = \frac{t^2}{2} \cdot \vec{a} + t\cdot old(\vec{v}) + old(\vec{p})}$$
This states that if the AMV is travelling in a straight line, then its position can be obtained through integration of $\dot{\vec{p}}$. By Theorem~\ref{thm:structp}, collinearity is also trivially an invariant of the $LRE$, since it does not modify any continuous variables.

\begin{figure}
    \centering
    \includegraphics[width=0.9\linewidth]{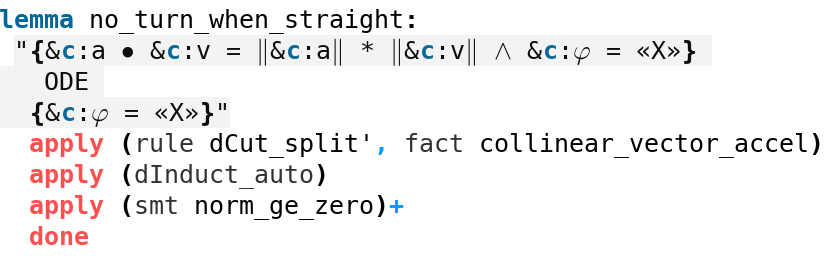}
    
    \vspace{-1ex}    
    \caption{Example proof in \isabelledH}
    \label{fig:proof_ex}
    
    \vspace{-3ex}
\end{figure}

We show proof of a further corollary in Figure~\ref{fig:proof_ex} in Isabelle/UTP: if the AMV is moving straight, then the 
heading is constant. We introduce a ghost variable $X$ for the current heading $\phi$. The proof proceeds by performing a differential cut (\textit{dCut\_split')}, which allows us to assume collinearity in the dynamical system. Theorem~\ref{thm:collinear} corresponds to the fact \textit{collinear\_vector\_accel} in Isabelle. Then, we use differential induction via the \textit{dInduct\_auto} tactic, which also applies algebraic simplification laws. Finally, we call \emph{sledgehammer} which provides SMT proofs to discharge the remaining proof obligation, which is essentially
$\vec{a} \cdot \vec{v} = \norm{\vec{a}}\!\cdot\!\norm{\vec{v}} \Rightarrow \dot{\phi} = 0$. This technique allows us to harness all the mathematical results proved in \textsf{HOL} and \textsf{HOL-Analysis} in our proofs~\cite{Harrison2005-Euclidean,Immler2012}.

Collinearity is established by $AP$ when the heading set point is the same as the actual heading:

\begin{theorem}[Autopilot Collinearity] \isalink{https://github.com/isabelle-utp/utp-main/blob/2152a303f6a517f125a44c7709010b00e6e2b0a8/theories/hyprog/examples/AMV.thy\#L369}
$$\hoaretriple{
    \begin{array}{c}
    0 \le s \land s \le rs \\
    \land \norm{rh - \phi} < \phi_{\epsilon} \\
    \land \vec{v} = s \cdot \begin{bmatrix} \sin(\phi) \\ \cos(\phi) \end{bmatrix}
    \end{array}}{AP}{\vec{a} \cdot \vec{v} = \norm{\vec{a}} \cdot \norm{\vec{v}}}$$
\end{theorem}
\begin{proof}
Hoare logic reasoning and vector arithmetic. A crucial fact is that $[\sin(\phi), \cos(\phi)] \cdot [\sin(\phi), \cos(\phi)] = 1$.
\end{proof}

\noindent The theorem shows that if the linear speed is between 0 and $rs$, $\phi$ and $rh$ are sufficiently close, and $\vec{v}$ can be derived from the linear speed and heading, then afterwards $\vec{a}$ and $\vec{v}$ are again collinear. Now, by the sequential composition law, we can compose this with \ref{thm:collinear} to obtain the same Hoare triple for $AP \mathrel{\hbox{\rm;}} Dyn$. Finally, we show a property of the LRE. \isalink{https://github.com/isabelle-utp/utp-main/blob/2152a303f6a517f125a44c7709010b00e6e2b0a8/theories/hyprog/examples/AMV.thy\#L419}
$$\hoaretriple{
    \!\begin{array}{c}
        m = MOM \land  \\
        \norm{\textit{ang}(\vec{wp} - \vec{p})-\phi} \le \phi_\epsilon \\
        \land (\exists \vec{o}\in ob. \norm{\vec{p} - \vec{o}} \le D)
    \end{array}\!}
    {LRE}{
    \!\begin{array}{c}
        m = HCM \\
        \land rs = H \land \\
         \norm{rh-\phi} \le \phi_\epsilon
    \end{array}\!}$$
This shows that if the LRE is in MOM, and the heading is currently towards the waypoint, but an obstacle is close, then it will transition to HCM, drop the set point speed to $H$ and the requested heading remains close to the actual heading.

\section{Discussion and Conclusion}
\label{sec:conclusion}

In this paper, we have made preliminary steps to integrating \ac{NC} with theorem proving in Isabelle. Octave and Isabelle's approaches to mathematics are, in many ways, quite different. Octave is focused on efficient \ac{NC}, whereas Isabelle is based on foundational mathematics and proof. %
Nevertheless, our investigation indicates that they can effectively be used together.

Most of the required Octave functions, such as, $\sin$, $\mathrm{sgn}$, and the vector operations are present in Isabelle, and are accompanied by a large body of theorems~\cite{Harrison2005-Euclidean,Immler2012,Immler2014}. The Archive of Formal Proofs\footnote{Archive of Formal Proofs. \url{https://www.isa-afp.org/}.} (AFP) has several useful libraries; for example, we used Manuel Eberl's library for calculating angles~\cite{Triangle-AFP}. Combining libraries with the flexible syntax of Isabelle, the program notation of Isabelle/UTP, and our matrix syntax, we can achieve a fairly direct translation of Octave functions, as \Cref{fig:AP-Isabelle} illustrates. Verification can be automated by the \textit{hoare-auto} tactic, though this dependends on arithmetic lemma libraries, some of which we needed to prove manually for the verification. Nevertheless, in our experience, \textit{sledgehammer}~\cite{Blanchette2011} performs quite well with arithmetic problems. Moreover, Isabelle has the approximation tactic, which can prove real and transcendental inequalities~\cite{Holzl2009-Approximate}. 

For the dynamics, it is necessary to produce an explicit system of first order ODEs, as shown in Definition~\ref{def:amv-dyn}. Consequently, any algebraic equations must be converted. For example, we could not include the value of $\phi$, but needed to give its derivative and include an axiom linking this with $s$ and $\vec{v}$. The challenge is finding invariants, and having sufficient background lemmas to prove the verification conditions.

There have been previous works on integrating \ac{NC} with deductive verification. Notably, Zhan et al. ~\cite{Zou2013-HHL} have used Hybrid CSP and an accompanying Hoare logic to verify Simulink block diagrams in Isabelle. Our work is more modest, in that we focus on sequential hybrid programs, but with a transparent translation and a high degree of automation. The dominant and most automated tool for hybrid systems deductive verification remains KeYmaera X~\cite{KeYmaeraX,Mitsch2017Obstacle}. Nevertheless, we believe that our preliminary results show the advantages of targeting Isabelle. Firstly, this allows integration of a variety of mathematical libraries to support reasoning. Secondly, we can combine notations with ODEs, and in the future aim to support refinement to code using libraries like Isabelle/C~\cite{Tuong2019-CIsabelle}. Thirdly, as illustrated by the obstacle register, with Isabelle/HOL we have the potential to extend \dL with additional features like collections as used in quantified \dL~\cite{Platzer2010-qDL}.

\bibliographystyle{IEEEtran}
\bibliography{}

\begin{thebibliography}{10}
\providecommand{\url}[1]{#1}
\csname url@samestyle\endcsname
\providecommand{\newblock}{\relax}
\providecommand{\bibinfo}[2]{#2}
\providecommand{\BIBentrySTDinterwordspacing}{\spaceskip=0pt\relax}
\providecommand{\BIBentryALTinterwordstretchfactor}{4}
\providecommand{\BIBentryALTinterwordspacing}{\spaceskip=\fontdimen2\font plus
\BIBentryALTinterwordstretchfactor\fontdimen3\font minus
  \fontdimen4\font\relax}
\providecommand{\BIBforeignlanguage}[2]{{%
\expandafter\ifx\csname l@#1\endcsname\relax
\typeout{** WARNING: IEEEtran.bst: No hyphenation pattern has been}%
\typeout{** loaded for the language `#1'. Using the pattern for}%
\typeout{** the default language instead.}%
\else
\language=\csname l@#1\endcsname
\fi
#2}}
\providecommand{\BIBdecl}{\relax}
\BIBdecl

\bibitem{Luckcuck2018}
M.~Farrell, M.~Luckcuck, and M.~Fisher, ``Robotics and integrated formal
  methods: Necessity meets opportunity,'' in \emph{Proc. 14th. Intl. Conf. on
  Integrated Formal Methods (iFM)}, vol. LNCS 11023.\hskip 1em plus 0.5em minus
  0.4em\relax Springer, 2018.

\bibitem{Gleirscher2018-NewOpportunitiesIntegrated}
M.~Gleirscher, S.~Foster, and J.~Woodcock, ``New opportunities for integrated
  formal methods,'' \emph{ACM Comput. Surv.}, vol.~52, no.~6, 2019.

\bibitem{Platzer2008}
A.~Platzer, ``Differential dynamic logic for hybrid systems,'' \emph{J. Autom
  Reasoning}, vol.~41, pp. 143--189, June 2008.

\bibitem{Isabelle}
T.~Nipkow, M.~Wenzel, and L.~C. Paulson, \emph{{Isabelle/HOL: A Proof Assistant
  for Higher-Order Logic}}.\hskip 1em plus 0.5em minus 0.4em\relax Springer,
  2002, vol. LNCS 2283.

\bibitem{Foster2020-dL}
J.~H.~Y. Munive, G.~Struth, and S.~Foster, ``Differential {Hoare} logics and
  refinement calculi for hybrid systems with {Isabelle/HOL},'' in
  \emph{RAMiCS}, ser. LNCS, vol. 12062.\hskip 1em plus 0.5em minus 0.4em\relax
  Springer, April 2020.

\bibitem{Mitsch2017Obstacle}
S.~Mitsch, K.~Ghorbal, D.~Vogelbacher, and A.~Platzer, ``Formal verification of
  obstacle avoidance and navigation of ground robots,'' \emph{International
  Journal of Robotics Research}, vol.~36, no.~12, 2017.

\bibitem{KeYmaeraX}
N.~Fulton, S.~Mitsch, J.-D. Quesel, M.~V\"{o}lp, and A.~Platzer, ``{KeYmaera
  X}: An axiomatic tactical theorem prover for hybrid systems,'' in
  \emph{CADE-25}, ser. LNCS, vol. 9195.\hskip 1em plus 0.5em minus 0.4em\relax
  Springer, 2015, pp. 527--538.

\bibitem{Blanchette2011}
J.~Blanchette, L.~Bulwahn, and T.~Nipkow, ``Automatic proof and disproof in
  {Isabelle/HOL},'' in \emph{FroCoS}, ser. LNCS, vol. 6989.\hskip 1em plus
  0.5em minus 0.4em\relax Springer, 2011.

\bibitem{Harrison2005-Euclidean}
J.~Harrison, ``A {HOL} theory of {E}uclidean space,'' in \emph{TPHOLs}, ser.
  LNCS, J.~Hurd and T.~Melham, Eds., vol. 3603.\hskip 1em plus 0.5em minus
  0.4em\relax Oxford, UK: Springer, 2005.

\bibitem{Immler2012}
F.~Immler and J.~H\"{o}lzl, ``Numerical analysis of {Ordinary Differential
  Equations} in {Isabelle/HOL},'' in \emph{3rd Intl. Conf. on Interactive
  Theorem Proving (ITP)}, vol. LNCS 7406.\hskip 1em plus 0.5em minus
  0.4em\relax Springer, 2012, pp. 377 -- 392.

\bibitem{Immler2014}
F.~Immler, ``Formally verified computation of enclosures of solutions of
  {Ordinary Differential Equations},'' in \emph{Proc. 6th NASA Formal Methods
  Symposium (NFM)}, ser. LNCS, vol. 8430.\hskip 1em plus 0.5em minus
  0.4em\relax Springer, 2014.

\bibitem{Foster2020-IsabelleUTP}
S.~Foster, J.~Baxter, A.~Cavalcanti, J.~Woodcock, and F.~Zeyda, ``Unifying
  semantic foundations for automated verification tools in {Isabelle/UTP},''
  \emph{Accepted for Science of Computer Programming}, 2020 (To Appear).

\bibitem{Hoare&98}
C.~A.~R. Hoare and J.~He, \emph{Unifying {Theories} of {Programming}}.\hskip
  1em plus 0.5em minus 0.4em\relax Prentice-Hall, 1998.

\bibitem{Munive2018-DDL}
J.~J. Huerta~y Munive and G.~Struth, ``Verifying hybrid systems with modal
  {Kleene} algebra,'' in \emph{RAMICS}, ser. LNCS, vol. 11194.\hskip 1em plus
  0.5em minus 0.4em\relax Springer, October 2018.

\bibitem{Foster2020AUV}
S.~Foster, Y.~Nemouchi, C.~O'Halloran, N.~Tudor, and K.~Stephenson, ``Formal
  model-based assurance cases in {Isabelle/SACM}: An autonomous underwater
  vehicle case study,'' in \emph{Proc. 8th Intl. Conf on Formal Methods in
  Software Engineering (FormaliSE)}.\hskip 1em plus 0.5em minus 0.4em\relax
  ACM, 2020.

\bibitem{Gleirscher2019-SEFM}
M.~Gleirscher, S.~Foster, and Y.~Nemouchi, ``Evolution of formal model-based
  assurance cases for autonomous robots,'' in \emph{SEFM}, ser. LNCS
  11724.\hskip 1em plus 0.5em minus 0.4em\relax Springer, 2019, pp. 87--104.

\bibitem{Ghorbal2014-AlgInv}
K.~Ghorbal and A.~Platzer, ``Characterizing algebraic invariants by
  differential radical invariants,'' in \emph{TACAS}, vol. LNCS 8413.\hskip 1em
  plus 0.5em minus 0.4em\relax Springer, 2014.

\bibitem{Foster07}
J.~Foster, M.~Greenwald, J.~Moore, B.~Pierce, and A.~Schmitt, ``Combinators for
  bidirectional tree transformations: A linguistic approach to the view-update
  problem,'' \emph{ACM Trans. Program. Lang. Syst.}, vol.~29, no.~3, May 2007.

\bibitem{Triangle-AFP}
M.~Eberl, ``Basic geometric properties of triangles,'' \emph{Archive of Formal
  Proofs}, Dec. 2015, \url{http://isa-afp.org/entries/Triangle.html}.

\bibitem{Holzl2009-Approximate}
J.~H\"{o}lzl, ``Proving inequalities over reals with computation in
  {Isabelle/HOL},'' in \emph{PLMMS}.\hskip 1em plus 0.5em minus 0.4em\relax
  ACM, August 2009.

\bibitem{Zou2013-HHL}
L.~Zou, N.~Zhan, S.~Wang, M.~Fr\"{a}nzle, and S.~Qin, ``Verifying {Simulink}
  diagrams via a {Hybrid Hoare Logic Prover},'' in \emph{Proc. Intl. Conf on
  Embedded Systems (EMSOFT)}.\hskip 1em plus 0.5em minus 0.4em\relax IEEE,
  2013.

\bibitem{Tuong2019-CIsabelle}
F.~Tuong and B.~Wolff, ``Deeply integrating {C11} code support into
  {Isabelle/PIDE},'' in \emph{Formal Integrated Development Environment
  (F-IDE)}, ser. EPTCS, vol. 310, 2019, pp. 13--28.

\bibitem{Platzer2010-qDL}
A.~Platzer, ``Quantified differential dynamic logic for distributed hybrid
  systems,'' in \emph{CSL}, ser. LNCS, vol. 6247.\hskip 1em plus 0.5em minus
  0.4em\relax Springer, 2010, pp. 469--483.

\end{thebibliography}
\end{document}